\documentclass[a4paper,11pt]{article}

\usepackage{authblk} 
\usepackage[labelfont=bf,labelsep=period]{caption}
\usepackage[nodayofweek]{datetime}
\usepackage{enumitem}
\usepackage{thmtools}
\usepackage{float}
\usepackage{graphicx}
\usepackage{hyperref}
\usepackage[utf8]{inputenc}
\usepackage{numbertabbing}
\usepackage{verbatim}
\usepackage{dsfont} 
\usepackage{url}
\usepackage{xcolor}
\usepackage{xspace}
\usepackage{tikz}
\usepackage{graphicx}
\usetikzlibrary{trees, positioning}

\usepackage{amsmath} 
\allowdisplaybreaks[2]          
\usepackage{amssymb} 
\usepackage{amsthm} 

\newdateformat{simple}{\THEDAY\ \monthname[\THEMONTH]\ \THEYEAR}
\simple

\floatstyle{ruled}
\newfloat{algo}{htbp}{algo}
\floatname{algo}{Algorithm}

\def \ifempty#1{\def\temp{#1} \ifx\temp\empty }
\newcommand{\str}[1]{\textsc{#1}}
\newcommand{\var}[1]{\textit{#1}}
\newcommand{\op}[1]{\textsl{#1}}
\newcommand{\msg}[2]{\ensuremath{\ifempty{#2} [\str{#1}] \else [\str{#1}, {#2}] \fi}}

\newcommand{\CF}{\ensuremath{\mathcal{F}}\xspace}
\newcommand{\CG}{\ensuremath{\mathcal{G}}\xspace}

\newcommand{\CK}{\ensuremath{\mathcal{K}}\xspace}

\newcommand{\CP}{\ensuremath{\mathcal{P}}\xspace}
\newcommand{\CQ}{\ensuremath{\mathcal{Q}}\xspace}

\newcommand{\CT}{\ensuremath{\mathcal{T}}\xspace}

\newcommand{\BF}{\ensuremath{\mathbb{F}}\xspace}

\newcommand{\BQ}{\ensuremath{\mathbb{Q}}\xspace}

\newcommand{\nil}{\ensuremath{\bot}}
\newcommand{\false}{\textsc{false}\xspace}
\newcommand{\true}{\textsc{true}\xspace}
\newcommand{\etal}{\emph{et al.}}
\newcommand{\zo}{\{0,1\}}
\newcommand{\propose}{send}
\newcommand{\proposed}{sent}

\theoremstyle{definition}
\newtheorem{definition}{Definition}
\theoremstyle{lemma}
\newtheorem{lemma}{Lemma}
\theoremstyle{theorem}
\newtheorem{theorem}{Theorem}
\theoremstyle{corollary}
\newtheorem{corollary}{Corollary}

\title{\bf Weaker Assumptions for Asymmetric Trust}

\author{Ignacio Amores-Sesar\thanks{Aarhus University
\texttt{amores-sesar@cs.au.dk}},
Christian Cachin\thanks{University of Bern
\texttt{christian.cachin@unibe.ch}},
Simon Holmgaard Kamp\thanks{Ruhr University Bochum
\texttt{simonhkamp@gmail.com}},
Juan Villacis\thanks{University of Bern
\texttt{juan.villacis@unibe.ch}}}

\date{}

\begin{document}
\maketitle

\begin{abstract}
In distributed systems with asymmetric trust, each participant is free to make its own trust assumptions about others, captured by an asymmetric quorum system. This contrasts with ordinary, symmetric quorum systems and threshold models, where trust assumptions are uniformly shared among participants. Fundamental problems like reliable broadcast and consensus are unsolvable in the asymmetric model if quorum systems satisfy only the classical properties of consistency and availability. Existing approaches overcome this by introducing stronger assumptions. We show that some of these assumptions are overly restrictive, so much so that they effectively eliminate the benefits of asymmetric trust. To address this, we propose a new approach to characterize asymmetric problems and, building upon it, present algorithms for reliable broadcast and consensus that require weaker assumptions than previous solutions. Our methods are general and can be extended to other core problems in systems with asymmetric trust.
\end{abstract}

\section{Introduction}
Asymmetric trust models, such as those proposed by Damgård~\etal~\cite{DBLP:conf/asiacrypt/DamgardDFN07}, Alpos \etal~\cite{DBLP:journals/dc/AlposCTZ24}, and Losa~\etal~\cite{DBLP:conf/wdag/LosaGM19} allow the development of distributed protocols in which each participant can operate under its own trust settings. These algorithms are built on top of asymmetric quorum systems, where each process independently defines its own quorums. The traditional consistency and availability properties must be satisfied by these systems, but as was shown by Li~\etal~\cite{DBLP:conf/wdag/LiCL23}, they are not enough to solve fundamental problems like reliable broadcast and consensus. Therefore, the models introduce additional assumptions on the quorums. The protocols for reliable broadcast and consensus of Alpos~\etal~\cite{DBLP:journals/dc/AlposCTZ24} require the existence of a \textit{guild}, a set of processes that contains a quorum for each of its members. Li~\etal~\cite{DBLP:conf/wdag/LiCL23} and Losa~\etal~\cite{DBLP:conf/wdag/LosaGM19} identify similar conditions. These assumptions can be very strong and restrictive, and while they are sufficient to provide solutions to the aforementioned problems, we will show that they are not necessary. This leads to asymmetric algorithms that work under weaker assumptions, making such systems more flexible and usable. 

We start by showing that current assumptions for algorithms such as reliable broadcast or consensus are so restrictive that it becomes unnecessary to use asymmetric trust models altogether. That is, given an asymmetric trust assumption that satisfies such requirements, we can derive an equivalent symmetric trust assumption such that no process will be worse off by adopting it. This result is surprising, considering that Alpos~\etal~\cite{DBLP:journals/dc/AlposCTZ24} show that asymmetric trust is more expressive than symmetric trust. 

To solve this, we introduce the notion of \emph{depth} of a process and propose a new approach to characterize asymmetric problems based on it. We use this concept to provide more general definitions of unauthenticated, asymmetric, asynchronous, Byzantine reliable broadcast and binary consensus. We also propose new algorithms for these problems that forego the \emph{guild} requirement in favor of weaker assumptions on~\emph{depth}.

The remainder of this document is organized as follows. Section~\ref{sec:relatedwork} discusses related work in the asymmetric trust setting. Section~\ref{sec:preliminaries} introduces the model and reviews key concepts that
serve as the foundation for our work. We show in Section~\ref{sec:degeneracy} that the current assumptions are so strong that they invalidate the benefits of asymmetric trust. Section~\ref{sec:depth} introduces the concept of depth. Sections~\ref{sec:rb} and~\ref{sec:consensus} present algorithms for reliable broadcast and consensus that forego the guild assumption in favor of weaker depth assumptions. Finally, conclusions are drawn in Section~\ref{sec:conclusion}. 

\section{Related Work}
\label{sec:relatedwork}
Threshold and symmetric quorum systems are fundamental tools in distributed computing, commonly used to establish the safety and liveness guarantees of various protocols~\cite{DBLP:journals/siamcomp/NaorW98,DBLP:journals/dc/MalkhiR98,DBLP:books/daglib/0017536,DBLP:books/daglib/0025983}. They constitute a crucial component of many systems in cloud computing~\cite{DBLP:journals/sigops/LakshmanM10, DBLP:conf/usenix/HuntKJR10} and cryptocurrencies~\cite{DBLP:conf/sosp/GiladHMVZ17, DBLP:journals/corr/abs-1807-04938, buterin2018ethereum2}. One drawback is that all participants must use the same quorum system, which restricts their freedom to choose who to trust. The asymmetric model, introduced by Damg{\aa}rd~\etal~\cite{DBLP:conf/asiacrypt/DamgardDFN07} and further developed by Alpos~\etal~\cite{DBLP:journals/dc/AlposCTZ24}, represents a novel approach that generalizes the symmetric paradigm. It allows participants to make trust choices of their own, reflecting social connections or other information available to them outside the protocol. Each participant can specify its own quorum system, giving it the freedom to have its own trust assumptions. Amores-Sesar~\etal~\cite{DBLP:conf/podc/Amores-SesarCVZ25} show that algorithms relying on symmetric quorums cannot be applied directly to the asymmetric model, as certain properties do not hold if the underlying quorums are asymmetric. This drives the need to adapt existing symmetric protocols and redefine protocol properties to fit the new model.

Existing work~\cite{DBLP:journals/dc/AlposCTZ24} introduces asymmetric versions of relevant primitives like reliable broadcast, binary consensus, and common coin. Amores-Sesar~\etal~\cite{DBLP:conf/podc/Amores-SesarCVZ25} extend this work by developing a DAG-based consensus protocol that relies on such primitives. Losa~\etal~\cite{DBLP:conf/wdag/LosaGM19} present another approach to model asymmetric trust. They forego the notion of fail-prone sets and strengthen the definition of quorums, by requiring each quorum to contain a quorum for each member. Li~\etal~\cite{DBLP:conf/wdag/LiCL23} expand on this model and address the topic of which quorum properties are necessary or sufficient to solve core asymmetric problems such as consensus. Sheff~\etal~\cite{DBLP:conf/opodis/SheffWRM20} propose a variant of the Paxos protocol that incorporates heterogeneous trust assumptions. 
Senn~and~Cachin~\cite{DBLP:conf/papoc/SennC25} explore asymmetric trust in the crash-fault scenario and show that, in this setting, asymmetric trust reduces to symmetric trust.

Asymmetric trust introduces novel challenges to the design of distributed protocols. One significant hurdle is that of processes that have \textit{wrong friends}, that is, the existence of processes that are correct and behave according to the underlying algorithm, but whose trust assumptions do not capture the actual faulty processes. Alpos~\etal~\cite{DBLP:journals/dc/AlposCTZ24} use the term \textit{naive} to refer to such participants and point out that their presence can lead to safety and liveness issues in asymmetric protocols. To address this, they introduce the concept of a guild, a set of processes that contains a quorum for each member, and provide algorithms for reliable broadcast and consensus that work correctly only in executions that have this structure. They point out that having a quorum system that only satisfies the consistency and availability properties can lead to liveness issues for some processes. In a slightly different model, Li~\etal~\cite{DBLP:conf/wdag/LiCL23} prove that the traditional properties of quorum consistency and availability are not sufficient to solve reliable broadcast and consensus. They introduce a new property, called strong availability, very similar to the notion of a guild, and prove that this property is sufficient to solve the aforementioned problems. They do not comment on the necessity of such a property.

Asymmetric consensus protocols have also transcended into the blockchain space with XRP~Ledger~\cite{DBLP:journals/corr/abs-1802-07242}(\url{https://xrpl.org}) and Stellar~\cite{mazieres2015stellar,DBLP:conf/sosp/LokhavaLMHBGJMM19}(\url{https://stellar.org}) as main representatives. In the XRP Ledger, each participant declares its own assumptions by listing other participating nodes from which it considers votes \cite{DBLP:journals/corr/abs-1802-07242,DBLP:conf/opodis/Amores-SesarCM20}. Stellar uses a similar approach, in which each participant keeps a list of participants it considers important and waits for a suitable majority of them to agree on a transaction before considering it settled~\cite{mazieres2015stellar,DBLP:conf/wdag/LosaGM19,DBLP:conf/sosp/LokhavaLMHBGJMM19}.

\section{Preliminaries}
\label{sec:preliminaries}
\subsection{Model}

All algorithms we propose are for the asynchronous unauthenticated setting. We will consider a system of $n$ \emph{processes} $\mathcal{P}=\{p_1, \dots, p_n\}$ that interact with each other by exchanging messages. A protocol for $\mathcal{P}$ consists of a collection of programs with instructions for all processes. They are presented using the event-based notation of Cachin~\etal~\cite{DBLP:books/daglib/0025983}. An execution starts with all processes in a special initial state; subsequently the processes repeatedly trigger events, react to events, and change their state through computation steps.  A process that follows its protocol during an execution is called \textit{correct}. A \textit{faulty} process, also called \textit{Byzantine}, may crash or deviate arbitrarily from its specification. We will assume that there is a low-level functionality for sending messages over point-to-point links between each pair of processes. This functionality is accessed through the events of \textit{sending} and \textit{receiving} a message. Point-to-point messages are authenticated and delivered reliably among correct processes. Here and from now on, the notation $\mathcal{A}^*$ for a system $\mathcal{A} \subseteq 2^\mathcal{P}$ , denotes the collection of all subsets of the sets in $\mathcal{A}$, that is, $\mathcal{A}^* = \{A' | A' \subseteq A, A \in \mathcal{A}\}$.

\subsection{Asymmetric Trust Overview}
\label{sec:context}
We consider and expand the asymmetric trust model proposed by Alpos~\etal~\cite{DBLP:journals/dc/AlposCTZ24}.
We summarize here the necessary context for understanding this work; we refer the reader to their paper for a full presentation. In protocols with asymmetric trust, each participant is free to make its own individual trust assumptions about others, captured by an asymmetric quorum system. This contrasts with ordinary, symmetric and threshold quorum systems, where all participants share the same trust assumptions. Given a set of processes $\mathcal{P}$, an \emph{asymmetric fail-prone system} $\mathbb{F} = [\mathcal{F}_1, \dots, \mathcal{F}_n]$, where $\mathcal{F}_i$ represents the failure assumptions  of process $p_i$, captures the heterogeneous model. Each $\mathcal{F}_i$ is a collection of subsets of $\mathcal{P}$ such that some $F \in \mathcal{F}_i$ with $F \subseteq \mathcal{P}$ is called a fail-prone set for $p_i$ and contains all processes that, according to $p_i$, may at most fail together in some execution \cite{DBLP:conf/asiacrypt/DamgardDFN07}. We can, in turn, proceed to define asymmetric Byzantine quorum systems, denoted by $\mathbb{Q}$. 
\begin{definition}
\label{def:abqs}
  An \emph{asymmetric Byzantine quorum system} $\mathbb{Q}$ for $\mathds{F}$ is an array of collections of sets $\mathbb{Q} = [\mathcal{Q}_1, \cdots, \mathcal{Q}_n]$ where $\mathcal{Q}_i \subseteq 2^{\mathcal{P}}$ for $i \in [1, n]$. The set $\mathcal{Q}_i \subseteq 2^{\mathcal{P}}$ is a symmetric quorum system of $p_i$ and any set $Q_i \in \mathcal{Q}_i$ is called a quorum for $p_i$. The system $\mathbb{Q}$ must satisfy the following two properties.

\begin{description}
    \item[Consistency:] The intersection of two quorums for any two processes contains at least one process for which either process assumes that it is not faulty, i.e.,
    $$
    \forall i, j \in [1, n], \forall Q_i \in \mathcal{Q}_i, \forall Q_j \in \mathcal{Q}_j, \forall F_{ij} \in \mathcal{F}_i^* \cap \mathcal{F}_j^*: Q_i \cap Q_j \nsubseteq F_{ij}.
    $$

    \item[Availability:] For any process $p_i$ and any set of processes that may fail together according to $p_i$, there exists a disjoint quorum for $p_i$ in $\mathcal{Q}_i$, i.e.,
    $$
        \forall i \in [1, n], \forall F_i \in \mathcal{F}_i : \exists Q_i \in \mathcal{Q}_i : F_i \cap Q_i = \emptyset.
    $$
\end{description}    
\end{definition}

Given an asymmetric quorum system $\mathbb{Q}$ for $\mathbb{F}$, an \emph{asymmetric kernel system} for $\mathbb{Q}$ is defined analogously as the array $\mathbb{K} = [\mathcal{K}_1, \dots, \mathcal{K}_n]$ that consists of the kernel systems for all processes in $\mathcal{P}$. A set $K_i \in \mathcal{K}_i$ is called a \emph{kernel} for $p_i$ and for each $K_i$ it holds that $\forall Q_i \in \mathcal{Q}_i, \ K_i \cap Q_i \neq \emptyset $, that is, a kernel intersects all quorums of a process.

Given an asymmetric fail-prone system $\mathbb{F}$, there will exist a valid asymmetric quorum system for $\mathbb{F}$ if and only if $\mathbb{F}$ satisfies the $B^3$ condition~\cite{DBLP:conf/asiacrypt/DamgardDFN07, DBLP:journals/dc/AlposCTZ24}. This property is defined as follows.
\begin{definition}[$B^3$-condition]
\label{def:b3}
  An asymmetric fail-prone system \BF satisfies the
  \emph{$B^3$-condition}, abbreviated as $B^3(\BF)$, whenever it holds that
  \[
    \forall i,j \in [1,n],
    \forall F_i \in \CF_i, \forall F_j\in\CF_j,
    \forall F_{ij} \in {\CF_i}^*\cap{\CF_j}^*: \,
    \CP \not\subseteq F_i \cup F_j \cup F_{ij} 
  \]
\end{definition}

If $B^3(\mathbb{F})$ holds, then the \emph{canonical quorum system}, defined as the complement of the asymmetric fail-prone system, is a valid asymmetric quorum system. This is similar to what happens in the symmetric case, where if the fail-prone system $\mathcal{F}$ satisfies the $Q^3$ property, the existence of a valid symmetric quorum system $\mathcal{Q}$ is guaranteed~\cite{DBLP:journals/joc/HirtM00}.

During an execution, the set of processes that fail is denoted by $F$. The members of $F$ are unknown to the processes and can only be identified by an outside observer and the adversary.  A process $p_i$ correctly foresees $F$ if $F \in \mathcal{F}_i^*$, that is, $F$ is contained in one of its fail-prone sets. Based on this information, it is possible to classify processes in three categories.

\begin{description}
    \item[Faulty:] a faulty process, i.e., $p_i \in F$;
    \item[Naive:] a correct process $p_i$, i.e. $p_i \notin F$, where $F \notin \mathcal{F}^*_i$; or
    \item[Wise:] a correct process $p_i$, i.e. $p_i \notin F$, where $F \in \mathcal{F}^*_i$
\end{description}

Alpos~\etal~\cite{DBLP:journals/dc/AlposCTZ24} show that naive processes may disrupt the safety and liveness guarantees of some protocols. In order to formalize this notion, they introduced the concept of a \textit{guild} which is central for many of the algorithms they propose (e.g., reliable broadcast, binary consensus).
\begin{definition}
    
A guild is a set of wise processes that contains one quorum for each of its members. Formally, a guild $\mathcal{G}$ for $\mathds{F}$ and $\mathds{Q}$, for an execution with faulty processes $F$, is a set of processes that satisfies the following two properties.
\begin{description}
    \item[Wisdom:] $\mathcal{G}$ is a set of wise processes, that is,
    $$
    \forall p_i \in \mathcal{G}: F \in \mathcal{F}_i^*.
    $$

    \item[Closure:] $\mathcal{G}$ contains a quorum for each of its members, that is,
    $$
    \forall p_i \in \mathcal{G} : \exists Q_i \in \mathcal{Q}_i: Q_i \subseteq \mathcal{G}.
    $$
\end{description}
\end{definition}

Alpos~\etal~\cite{DBLP:journals/dc/AlposCTZ24} provide reliable broadcast and binary consensus implementations which function only in executions that contain a guild. In relation to this, Definition~\ref{def:toleratedsystem} presents the notion of a tolerated system, originally introduced by Alpos~\etal~\cite{DBLP:journals/dc/AlposCTZ24}. It is built on top of tolerated sets, which are sets that contain all elements outside of a guild. This portrays the resilience of the trust assumptions, as even when all processes in one of the tolerated sets fail, there is still a guild in the system. This structure is used by Alpos~\etal~\cite{DBLP:journals/dc/AlposCTZ24} and Zanolini~\cite{Zanolini2023} in their construction of several asymmetric protocols.

\begin{definition}[Tolerated system~\cite{DBLP:journals/dc/AlposCTZ24}]\label{def:toleratedsystem}
  Given an asymmetric Byzantine quorum system \BQ and an execution with
  faulty processes $F$, a set of processes~$T$ is called \emph{tolerated
    (by \BQ)} if a non-empty guild \CG for $F$ and \BQ exists such that
  $T = \CP \setminus \CG$.

  The \emph{tolerated system} \CT of an asymmetric Byzantine quorum system
  \BQ is the maximal collection of tolerated sets, where $F$ ranges over
  all possible executions.
\end{definition}

Lemma~\ref{lem:b3-q3}, which is crucial for the construction presented in Section~\ref{sec:degeneracy}, shows that if there is an asymmetric fail-prone system $\mathbb{F}$ such that $B^3(\mathbb{F})$ holds, then $Q^3(\CT)$ also holds. This means that $\CT$ is a valid symmetric fail-prone system~\cite{DBLP:journals/joc/HirtM00}.

\begin{lemma}[\cite{DBLP:journals/dc/AlposCTZ24}]\label{lem:b3-q3}
  Let $\mathbb{Q}$ be an asymmetric Byzantine quorum system among processes \CP with asymmetric fail-prone system
  $\mathbb{F} = \overline{\mathbb{Q}}$, i.e., such that $\mathbb{Q}$ is a canonical asymmetric Byzantine quorum system,
  and let \CT be the tolerated system of $\mathbb{Q}$. If $B^3(\mathbb{F})$ holds, then $\mathcal{T}$ is a valid symmetric fail-prone system, i.e.,  $Q^3(\CT)$ holds.
\end{lemma}

\subsection{Quorum-based Algorithms}
In the following sections, we will give results regarding the existence of algorithms to solve problems such as reliable broadcast and consensus. To do this, we must first define what we mean by an algorithm. We will consider all algorithms to be \textit{quorum-based algorithms}. Definition~\ref{def:qba} formalizes this notion. It is derived from the definitions used by Losa~\etal~\cite{DBLP:conf/wdag/LosaGM19} and Li~\etal~\cite{DBLP:conf/wdag/LiCL23} , but it is adapted to fit the asymmetric model of Alpos~\etal~\cite{DBLP:journals/dc/AlposCTZ24}. We use the term \emph{initiator of the execution} to refer to the sole process that triggers the execution of an algorithm. In our work, this only applies within the context of reliable broadcast, where the initiator is the process that broadcasts a value. In consensus, there is no single initiator because all processes propose a value. 

\begin{definition}[Quorum-based algorithm]
\label{def:qba}
  Given a distributed algorithm $A$, processes $\mathcal{P}$, and asymmetric quorum system $\mathbb{Q}$, we will say that $A$ is quorum-based if the following properties hold
  \begin{itemize}
      \item In any execution where a correct process $p_i$ issues a response, there is a quorum $Q_i \in \mathcal{Q}_i$ such that $p_i$ received a message from each member of $Q_i$. 
      \item If the state of a correct process $p_i$ changes upon receiving a message from some process $p_j$, then there is a $ Q_i \in \mathcal{Q}_i$ such that $p_j \in Q_i$ or $p_j$ is the initiator of the execution of $A$.
      \item For every execution where a correct process $p_i$ issues a response, there exists an execution where $p_i$ only receives messages from the members of one of its quorums $Q_i \in \mathcal{Q}_i$ and the initiator and $p_i$ issues the same response.
  \end{itemize}
\end{definition}

\section{Guild Assumptions are not Asymmetric}
\label{sec:degeneracy}

In the asymmetric trust setting, Li~\etal~\cite{DBLP:conf/wdag/LiCL23} show that a quorum system satisfying the classical properties of consistency and availability alone does not suffice to solve fundamental problems such as reliable broadcast or consensus. The heterogeneous views of the system make it necessary to assume more structure. Different approaches have reached similar conclusions about the assumptions needed to address these problems under asymmetric trust. Alpos~\etal~\cite{DBLP:journals/dc/AlposCTZ24} require the existence of a guild, while Li~\etal~\cite{DBLP:conf/wdag/LiCL23} propose that the system must satisfy the \emph{strong availability} property. This is closely related to the idea of a guild, as it also requires the existence of a well-behaved group of processes that includes a quorum for every participant. In this section, we examine the implications of these assumptions, focusing on the model proposed by Alpos~\etal~\cite{DBLP:journals/dc/AlposCTZ24}. 

Guilds are a strong assumption, essentially requiring all members to have common beliefs, which goes against the idea of asymmetric trust. We show that for any asymmetric Byzantine fail-prone system $\mathbb{F}$, if a guild exists in an execution, then we can build an equally valid symmetric trust assumption $\mathcal{F}$ from $\mathbb{F}$. This means that, in such cases, asymmetric trust reduces to symmetric trust.

In the crash-fault model, Senn and Cachin~\cite{DBLP:conf/papoc/SennC25} already show that given an asymmetric fail-prone system~$\mathbb{F}$, it is possible to construct a symmetric fail-prone system~$\mathcal{F}$ from $\mathbb{F}$ such that if all processes adopt $\mathcal{F}$ as their trust assumption, no process will be worse off. This result implies that there is no need for asymmetric trust in that setting. We extend this by showing that a similar scenario occurs in the Byzantine fault model for all executions with a guild. This effectively invalidates the advantages of asymmetric trust for the algorithms proposed by Alpos~\etal~\cite{DBLP:journals/dc/AlposCTZ24} which rely on the existence of a guild.

We now show how to construct a symmetric fail-prone system $\mathcal{F}$ from an asymmetric fail-prone system $\mathbb{F}$, such that in all executions with a guild, $\mathcal{F}$ is as expressive as $\mathbb{F}$. We will build it based on the notion of the tolerated system, presented in Section~\ref{sec:relatedwork} and used by Alpos~\etal~\cite{DBLP:journals/dc/AlposCTZ24} and Zanolini~\cite{Zanolini2023} in their consensus protocols.

\begin{definition}[Tolerated Symmetric Fail-Prone System]
\label{def:tsfps}
    Let $\mathbb{F}$ be an asymmetric fail-prone system that satisfies $B^3(\mathbb{F})$, $\mathbb{Q}$ be its associated canonical quorum system and \CT be the tolerated system derived from $\mathbb{Q}$ according to Definition~\ref{def:toleratedsystem}. We interpret $\mathcal{T}$ as a symmetric fail-prone set and denote it by $\hat{\mathcal{F}}$, that is,  
    $
    \hat{\mathcal{F}} = \mathcal{T}
    $, with a one-to-one correspondence between the tolerated sets and fail-prone sets. 
\end{definition}

From Lemma~\ref{lem:b3-q3} we know that the tolerated symmetric fail-prone system satisfies the $Q^3$ property and therefore there will be at least one valid symmetric Byzantine quorum system associated to it.

\begin{lemma}
\label{lem:guild-implies-sym}
  Let $\mathbb{F}$ be an asymmetric fail-prone system that satisfies $B^3(\mathbb{F})$, \CT be the tolerated system derived from $\mathbb{F}$ and $\hat{\mathcal{F}}$ be the tolerated symmetric fail-prone system derived according to Definition~\ref{def:tsfps}. Given an execution with a guild $\mathcal{G}$ and faulty processes $F$, then $F \in
  \hat{\mathcal{F}}$.
\end{lemma}

\begin{proof}
    By Definition~\ref{def:toleratedsystem}, there exists a tolerated set $T = \mathcal{P} \setminus \mathcal{G}$ such that $T \in \mathcal{T}$. Since $F \cap \mathcal{G} = \emptyset$ we know that $F \subseteq T$ and consequently there exists a fail-prone set $F' \in \hat{\mathcal{F}}$ such that $T=F'$. Therefore, the set of faulty processes will be considered by the fail-prone system $\hat{\mathcal{F}}$.
\end{proof}

\begin{theorem}
\label{th:guild-implies-sym}
  Let $\mathbb{F}$ be an asymmetric fail-prone system that satisfies $B^3(\mathbb{F})$. Given an execution with faulty processes $F$ and at least one guild $\mathcal{G}$, it is possible to derive a symmetric fail-prone system $\mathcal{F}$  from $\mathbb{F}$ such that $F \in \mathcal{F}^*$.
\end{theorem}

\begin{proof}
    We first identify the tolerated symmetric fail-prone system $\hat{\mathcal{F}}$ according to Definition~\ref{def:tsfps}. Since $B^3(\mathbb{F})$ we know that $Q^3(\hat{\mathcal{F}})$ and therefore it is a valid symmetric Byzantine fail-prone system. Applying Lemma~\ref{lem:guild-implies-sym}, we know that if we are in an execution with a guild $\mathcal{G}$ and faulty processes $F$, then $F \in \hat{\mathcal{F}}^*$. Therefore, by taking $\mathcal{F} = \hat{\mathcal{F}}$ it is possible to construct a symmetric fail-prone system $\mathcal{F}$ derived from $\mathbb{F}$ such that $F \in \mathcal{F}^*$.
\end{proof}

Theorem~\ref{th:guild-implies-sym} shows that for executions with a guild, every process is at least as well off using the symmetric quorum system $\mathcal{F}$ instead of $\mathbb{F}$. For a wise process $p_i$ the situation will remain the same regardless of the quorum system chosen. With the asymmetric system $F \in \mathcal{F}^*_i$ will hold, while with the symmetric alternative $F \in \mathcal{F}^*$. For naive processes the situation improves, as now the faulty processes will be considered in their trust assumption. Thus, there is no reason for a process not to adopt the derived symmetric fail-prone system $\mathcal{F}$. Building $\mathcal{F}$ can be computationally expensive because it requires computing $\mathcal{T}$. However, existing asymmetric consensus protocols~\cite{DBLP:journals/dc/AlposCTZ24, Zanolini2023} also require computing $\mathcal{T}$, so our approach is no more costly than current solutions.

Hence, if a guild is needed to solve a problem in the asymmetric setting, there exists a way to circumvent this and use any existing symmetric algorithm while obtaining the same guarantees that an asymmetric algorithm would provide. Therefore, if a problem can only be solved if a guild exists, the interest in using asymmetric trust to solve it decreases. This motivates the search for better algorithms for problems such as reliable broadcast and consensus, where the only known asymmetric algorithms require a guild. 

One limitation of this result is that it assumes knowledge of the trust assumptions for all processes ($\mathbb{F}$). Although this assumption holds in the Damgård \etal~\cite{DBLP:conf/asiacrypt/DamgardDFN07} and Alpos~\etal~\cite{DBLP:journals/dc/AlposCTZ24} models, it is not required in other asymmetric trust models~\cite{DBLP:conf/wdag/LosaGM19}. We also note that knowledge of the faulty processes $F$ is not required to derive the symmetric fail-prone system presented in Theorem~\ref{th:guild-implies-sym}.

\section{Depth}
\label{sec:depth}
In the symmetric world, processes are classified as correct or faulty. This suffices to express the properties of algorithms and capture all potential states of a process. In the asymmetric world this is not enough, a process might be correct but have incorrect trust assumptions. Damgård~\etal~\cite{DBLP:conf/asiacrypt/DamgardDFN07} and Alpos~\etal~\cite{DBLP:journals/dc/AlposCTZ24} address this problem by introducing the notions of naive, wise, and guilds.

Existing protocols for reliable broadcast and consensus in the Alpos~\etal~\cite{DBLP:journals/dc/AlposCTZ24} model require a guild to ensure safety and liveness. The results of Section~\ref{sec:degeneracy}, coupled with the inherent rigidity of having this kind of assumptions, motivate the search for new asymmetric algorithms that forego such assumptions. In order to do this, we introduce a new way to describe the \textit{quality} of processes, based on how much they can rely on other processes during an execution. We denote this concept as the \emph{depth} of a process. We then show how it can be used to characterize asymmetric problems in a more fine-grained manner. Definition~\ref{def:depth-process} introduces this notion. This approach allows to generalize the concept of \textit{wise process} or \textit{guild member} and introduces more expressivity to the system.

\begin{definition}[Depth of a process]
\label{def:depth-process}
    For an execution, we recursively define the notion of a correct process having depth $d$ as follows
    \begin{itemize}

        \item Any correct process $p_i$ has depth 0. 
    
        \item Additionally, a correct process $p_i$ has depth $d\geq 1$ if it has a quorum such that all processes contained in it have depth at least $d-1$, i.e.,
        $$
        \exists Q \in \mathcal{Q}_i, \forall p_j \in Q: \text{ $p_j$ is correct, has depth $s$, and $s\geq d-1$}
        $$
    \end{itemize}
\end{definition}

Note that a process with maximal depth $d$ also has depth $d'$ for all $0\leq d'\leq d$. We will focus on the \emph{maximal depth} of a process. In the following sections, any reference to the depth of a process should be understood as its maximal depth. Following the terminology of Damgård~\etal~\cite{DBLP:conf/asiacrypt/DamgardDFN07}, naive processes have depth 0, wise processes have depth at least 1, and processes in a guild have infinite depth.

In Sections~\ref{sec:rb}~and~\ref{sec:consensus} we show how the properties of protocols can be adapted to take into account the depth of processes. We modify them to show that some property must hold for all processes with depth at least $d$. For example, the \textit{totality} property of reliable broadcast is updated to state that all processes with depth $d$ eventually deliver a message.

\section{Reliable Broadcast}
\label{sec:rb}

Section~\ref{sec:degeneracy} shows that if an execution includes a guild, we can convert asymmetric trust assumptions into symmetric ones. This means that if a problem can only be solved in executions with a guild, the interest in using asymmetric trust to solve it decreases. So far, the only known protocols for reliable broadcast and consensus in the Alpos~\etal~\cite{DBLP:journals/dc/AlposCTZ24} model require the existence of a guild.  This motivates the search for new algorithms that work without relying on guild assumptions for these problems. We start this search with the Byzantine reliable broadcast problem. Our approach weakens the trust assumptions needed and eliminates the need for a guild. Definition~\ref{def:arbc} presents a way to characterize reliable broadcast based on the depth of processes. Its properties are specific to processes that have at least a certain maximal depth.  

\begin{definition}[Depth-characterized asymmetric Byzantine reliable broadcast]\label{def:arbc}
  A protocol for \emph{asymmetric Byzantine reliable broadcast} with sender~$p_s$ and depth~$d$, shortened as \op{RB[$d$]}, defined through the events \textit{dar-broadcast(m)} and \textit{dar-deliver(m)} satisfies the following properties:
\begin{itemize}
    \item  \textbf{Validity: }If a correct process $p_s$
  \op{dar-broadcasts} a message~$m$, then all processes with depth $d$
  eventually \op{dar-deliver}~$m$.
  \item \textbf{Consistency: }   If some process with depth $d$ \textit{dar-delivers} $m$ and another process with depth $d$ \textit{dar-delivers} $m'$, then $m=m'$.
    \item \textbf{Integrity: }  Every process with depth $d$ \textit{dar-delivers} $m$ at most once. Moreover, if the sender $p_s$ is correct and the receiver has depth $d$, then $m$ was previously \textit{dar-broadcast} by $p_s$.
    \item \textbf{Totality:} If a process with depth $d$
  \op{dar-delivers} some message, then all processes with depth $d$
  eventually \op{dar-deliver} a message.
\end{itemize}
\end{definition}

The reliable broadcast proposed by Alpos~\etal~\cite{DBLP:journals/dc/AlposCTZ24} is closely related to \op{RB[$\infty$]}, although not as general. \op{RB[$\infty$]} guarantees a solution for all processes with depth $\infty$, whereas the protocol by Alpos~\etal~\cite{DBLP:journals/dc/AlposCTZ24}  only guarantees a solution for those processes with depth $\infty$ that also belong to the maximal guild. That is, processes with depth $\infty$ that are members of a guild, but not of the maximal guild, are not guaranteed to receive a solution.

A natural question arises, which is the minimum depth needed to implement the different primitives. A low depth not only corresponds to weaker asymmetric assumptions, but also reflects, to a certain level, the communication complexity of an algorithm.

If an algorithm solves RB[$s$] then it also solves \op{RB[$s'$]} for all $s'\geq s$. This simplifies the search for a solution by reducing it to finding an algorithm that works for the minimal value of $s$. Additionally, there are no algorithms that can solve \op{RB[1]}, therefore, we must search for a protocol that solves the problem for processes with a depth of at least 2. These results are proven in Lemmas~\ref{lem:solve-s-solve-s'}~and~\ref{lem:norb1}.

Algorithm~\ref{alg:rb3} presents a solution for \op{RB[3]}. Every process waits to receive a quorum of \str{Readyafterecho} messages associated to the same message $m$ and round $r$ before delivering $m$. A process with depth 3 will only receive such a quorum if there exists a set of processes~$Q$ (which form a quorum for a wise process) that can attest that the sender indeed sent the value $m$. Since there will be an attesting quorum for every process that delivers, and since quorums for wise processes intersect in at least one correct process, we can deduce that all processes with depth 3 that deliver a value $m$ will deliver the same value. In addition, we use a technique similar to Bracha's kernel amplification~\cite{DBLP:journals/jacm/BrachaT85} to ensure that if a valid process delivers then all valid processes will deliver. The arrays in lines~\ref{line:srae}, ~\ref{line:srar}, ~\ref{line:rae}, and~\ref{line:rar} are hashmaps, so even though they are depicted as having infinite size they are actually sparsely populated.  The algorithm has a latency of 3 asynchronous rounds in the good case and 5 in the bad case. Its correctness is proved in Theorem~\ref{lem:rb3correct}. The existence of an algorithm to solve \op{RB[2]} remains an open question. 

\begin{algo*}[h!]
\vbox{
\small
\begin{numbertabbing}\reset
  xxxx\=xxxx\=xxxx\=xxxx\=xxxx\=xxxx\=MMMMMMMMMMMMMMMMMMM\=\kill
  \textbf{State} \label{}\\
  \> \(\var{sentecho} \gets \false\): indicates whether $p_i$ has sent \str{echo} \label{}\\
  \> \(\var{echos} \gets [\bot]^{n}\): collects the received \str{echo} messages from other processes \label{}\\
  \> \parbox[t]{0.87\linewidth}{ \(\var{sentrae} \gets [\false]^\infty\): indicates whether $p_i$ has sent \str{readyafterecho} in round n, a hashmap} \label{line:srae}\\
  \> \parbox[t]{0.87\linewidth}{ \(\var{sentrar} \gets [\false]^\infty\): indicates whether $p_i$ has sent \str{readyafterready} in round n, a hashmap} \label{line:srar}\\
  \> \parbox[t]{0.87\linewidth}{\(\var{readysafterecho} \gets [\bot, \bot]^{\infty\times n}\): collects \str{readyafterecho} messages from other processes, a hashmap} \label{line:rae}\\
  \> \parbox[t]{0.87\linewidth}{\(\var{readysafterready} \gets [\bot, \bot]^{\infty\times n}\): collects \str{readyafterready} messages from other processes, a hashmap} \label{line:rar}\\
  \> \(\var{delivered} \gets \false\): indicates whether $p_i$ has delivered a message\label{}\\
  \\
  \textbf{upon invocation} \(\op{dar-broadcast}(m)\) \textbf{do} \label{}\\
  \> send message \msg{send}{m} to all \(p_j \in \CP\) \label{}\\
  \\
  \textbf{upon} receiving a message \msg{send}{m} from $p_s$
    \textbf{such that} \(\neg \var{sentecho}\) \textbf{do} \label{}\\
  \> \(\var{sentecho} \gets \true\) \label{}\\
  \> send message \msg{echo}{m} to all \(p_j \in \CP\)\label{} \\
  \\
  \textbf{upon} receiving a message \msg{echo}{m} from \(p_j\) \textbf{do} \label{}\\
  \> \textbf{if} \(\var{echos}[j] = \bot\) \textbf{then} \label{}\\
  \> \> \(\var{echos}[j] \gets m\) \label{}\\
  \\
  \parbox[t]{0.87\linewidth}{ \textbf{upon exists} \(m \not= \bot \) \textbf{such that}
     \( \{p_j \in \CP | \var{echos}[j] = m\} \in \CQ_i \)
     \textbf{and}  \(\neg \var{sentrae[1]}\) \textbf{do}
     \` // a quorum for $p_i$} \label{}\\
  \> \(\var{sentrae[1]} \gets \true\) \label{}\\
  \> send message \msg{readyafterecho, 1}{m} to all $p_j \in \CP$ \label{}\\
  \\
  \parbox[t]{0.87\linewidth}{ \textbf{upon exists} \(m \not= \bot, r>0 \) \textbf{such that}
     \( \{p_j \in \CP | \var{readysafterecho}[r][j] = m\} \in \CK_i \)
     \textbf{and} \(\neg \var{sentrar[r]}\) \textbf{do}}
     \label{} \\
  \> \(\var{sentrar[n]} \gets \true\) \label{}\\
  \> send message \msg{readyafterready}{r, m} to all $p_j \in \CP$ \label{}\\
  \\
  \textbf{upon} receiving a message \msg{readyafterecho}{r, m} from \(p_j\)
     \textbf{do} \label{}\\
  \> \textbf{if} \(\var{readysafterecho}[r][j] = \bot\) \textbf{then} \label{}\\
  \> \> \(\var{readysafterecho}[r][j] \gets m\) \label{}\\
  \\

  \textbf{upon} receiving a message \msg{readyafterready}{r, m} from \(p_j\)
     \textbf{do} \label{}\\
  \> \textbf{if} \(\var{readysafterready}[r][j] = \bot\) \textbf{then} \label{}\\
  \> \> \(\var{readysafterready}[r][j] \gets m\) \label{}\\
  \\

  \parbox[t]{0.87\linewidth}{ \textbf{upon exists} \(m \not= \bot, r>1 \) \textbf{such that}
     \( \{p_j \in \CP | \var{readyafterready}[r][j] = m\} \in \CQ_i \)
     \textbf{and}  \(\neg \var{sentrae[r]}\) \textbf{do}}
    \label{}\\
  \> \(\var{sentrae[r+1]} \gets \true\) \label{}\\
  \> send message \msg{readyafterecho}{r+1, m} to all $p_j \in \CP$ \label{line:sendrae}\\
  \\

  \parbox[t]{0.87\linewidth}{ \textbf{upon exists} \(m \not= \bot, r>0 \) \textbf{such that} \( \{ p_j \in \CP | \var{readysafterecho}[r][j] = m\} \in \CQ_i \)  \textbf{and} \(\neg \var{delivered}\) \textbf{do}} \label{}\\
  \> \(\var{delivered} \gets \true\) \label{}\\
  \> \textbf{output} \(\op{dar-deliver(m)}\)\label{} \\ [-5ex]
\end{numbertabbing}
}
\caption{Depth 3 asymmetric reliable broadcast with sender~$p_s$ (\op{RB[3]})
  (process~$p_i$)}
\label{alg:rb3}
\end{algo*}

\begin{lemma}
    In any execution with at least one process $p_3$ with depth 3 and for any $n>1$, if a process $p_i$ with depth 1 sends a \msg{Readyafterecho}{n, m} message then there exists another process $p_j$ with depth 1 that sent a \msg{Readyafterecho}{n-1, m} message.
    \label{lem:d1echo-d1echo}
\end{lemma}
\begin{proof}
    According to the code of Algorithm~\ref{alg:rb3}, $p_i$ sends a \msg{Readyafterecho}{n, m} message after receiving \msg{Readyafterready}{n-1, m} messages from one of its quorums $Q_i \in \mathcal{Q}_i$. Process $p_3$ has a quorum $Q_2 \in \mathcal{Q}_3$ such that all its members have depth at least 2. By the quorum consistency property there is at least one process $p_2 \in Q_i \cap Q_2$. Therefore, $p_2$ sent a \msg{Readyafterready}{n-1, m} message. 

    According to Algorithm~\ref{alg:rb3}, $p_2$ does this  after receiving \msg{Readyafterecho}{n-1, m} messages from a kernel $K_2 \in \mathcal{K}_2$. Since $p_2$ has depth 2, it has at least one quorum where all its members have depth 1, and therefore each of its kernels contains at least one process with depth 1. Let $p_j$ be such a process in $K_2$. Since $p_j$ has depth 1 and it sent a \msg{Readyafterecho}{n-1, m} message we have proven the lemma. 
\end{proof}

\begin{lemma}
If a process $p_i$ with depth 1 sends a \msg{Readyafterecho}{1, m} message then $\exists Q_i \in \mathcal{Q}_i$  such that all processes in $Q_i$ sent a \msg{Echo}{m} message.
\label{lem:d1echo}
\end{lemma}
\begin{proof}
    According to Algorithm~\ref{alg:rb3}, a correct process only sends a \msg{Readyafterecho}{1, m} message after receiving \msg{Echo}{m} messages from one of its quorums $Q_i \in \mathcal{Q}_i$. This proves the lemma. 
\end{proof}

\begin{lemma}
    If a process $p_i$ with depth 3 delivers a value $m$ then there is a process $p_j$ with depth 1  such that $\exists Q_j \in \mathcal{Q}_j$  such that all processes in $Q_j$ sent a \msg{Echo}{m} message.
    \label{lem:d3-d1}
\end{lemma}
\begin{proof}
    Let $p_i$ be a process with depth 3 that delivered a value $m$, and $Q_2 \in \mathcal{Q}_i$ be the quorum of $p_i$ that contains only processes with depth at least 2. Its existence follows from the definition of depth.  In order to deliver the value $p_i$ received \msg{Readyafterecho}{n, m} messages from all processes in one of its quorums $Q_i \in \mathcal{Q}_i$ and some value $n$. Since $Q_2 \cap Q_i \neq \emptyset$ there is at least one process $p_j''$ with depth 2 that sent a \msg{Readyafterecho}{n, m} message. 

    Since a process with depth 2 also has depth 1, we can apply Lemma~\ref{lem:d1echo-d1echo} to obtain that there is another process $p_j'$ with depth 1 that sent a \msg{Readyafterecho}{n-1, m} message. By repeatedly applying Lemma~\ref{lem:d1echo-d1echo} we reach a a process $p_j$ with depth 1 that sent a \msg{Readyafterecho}{1, m} message. We now apply Lemma~\ref{lem:d1echo} to obtain that $\exists Q_j \in \mathcal{Q}_j$  such that all processes in $Q_j$ sent a \msg{Echo}{m} message. This proves the lemma. 
\end{proof}

\begin{lemma}
    If a quorum-based algorithm $A$ solves \op{RB[d]} then it solves \op{RB[d$'$]} for all $\var{d'}>\var{d}$
    \label{lem:solve-s-solve-s'}
\end{lemma}
\begin{proof}
    Suppose that an algorithm $A$ solves \op{RB[d]}. Denote by $S$ the set of all processes that have depth \var{d} and by $S'$ the set of all processes that have depth \var{d'}. Since $S'\subseteq S$ then all properties of \op{RB[d$'$]} will be satisfied for the processes in $S'$ after executing $A$. Therefore $A$ also solves \op{RB[d$'$]}.
    
\end{proof}

\begin{lemma}
    There is no quorum-based algorithm that solves \op{RB[1]} 
    \label{lem:norb1}
\end{lemma}
\begin{proof}
    Consider  a system with processes $p_1, p_2, p_3, p_4, p_5$ and $p_6$ and the following fail-prone system. 

    $$
    \mathbb{F}_D: \quad
    \begin{aligned}
    \mathcal{F}_1 &= \{\{1, 2, 5, 6\}, \{1, 2, 3\}\} \\
    \mathcal{F}_2 &= \{\{1, 2, 5, 6\}, \{1, 2, 3\}\} \\
    \mathcal{F}_3 &= \{\{1, 2, 4\}\} \\
    \mathcal{F}_4 &= \{\{1, 2, 3\}\} \\
    \mathcal{F}_5 &= \{\{1, 2, 4\}\} \\
    \mathcal{F}_6 &= \{\{1, 2, 4\}\} 
    \end{aligned}
    $$
    
    It satisfies the $B^3$ property and therefore the following canonical quorum system is valid.
    $$
    \mathbb{Q}_D: \quad
    \begin{aligned}
    \mathcal{Q}_1 &= \{\{3, 4\}, \{4, 5, 6\}\} \\
    \mathcal{Q}_2 &= \{\{3, 4\}, \{4, 5, 6\}\} \\
    \mathcal{Q}_3 &= \{\{3, 5, 6\}\} \\
    \mathcal{Q}_4 &= \{\{4, 5, 6\}\} \\
    \mathcal{Q}_5 &= \{\{3, 5, 6\}\} \\
    \mathcal{Q}_6 &= \{\{3, 5, 6\}\} 
    \end{aligned}
    $$

    Let the processes $p_1$, $p_2$, $p_3$, $p_4$ be correct and $p_5$, $p_6$ be Byzantine. Since $p_1$ and $p_2$ have a quorum composed only of correct processes ($\{3, 4\}$) they have depth 1. We now show an execution in which $p_1$ \op{dar-delivers} while $p_2$ does not, which breaks the totality property of \op{RB[1]}.

    Firstly consider an execution $E_0$ in which $p_5$ \op{dar-broadcasts} value 0, all processes behave according to the algorithm. Processes $p_1$ and $p_2$ will deliver value 0. From Definition~\ref{def:qba} we know that there is an equivalent execution $E$ where $p_1$ and $p_2$ deliver the same value and only receive messages from the processes in their quorum $\{4, 5, 6\}$ and $p_5$. 

    Now consider an execution $E'$ in which processes $p_5$ and $p_6$ behave towards $p_4$ like in $E$. Note that $p_4$ will behave like in $E$ since at every step it will receive the same messages from the same processes as it did in $E$. Let $p_5$ and $p_6$ also behave towards $p_2$ like in $E$ and send no messages towards all other processes. Process $p_3$ does not begin the execution since it never receives any message from the initiator (Byzantine) nor any member of one of its quorums (processes 3, 5 or 6).  Process $p_2$ will deliver value 0 as it receives the same messages from the same processes as in $E$. On the contrary, process $p_1$ will not deliver any value since it will never receive messages from all processes in one of its quorums, which would contradict Definition~\ref{def:qba}. This comes from the fact that $p_3$ will never send any message and that $p_5$ and $p_6$ are Byzantine and will not send it any message. 
    Therefore, $p_2$ delivers a value but $p_1$ does not, which breaks the totality property of \op{RB[1]}.
\end{proof}

\begin{theorem}
    Algorithm~\ref{alg:rb3}  solves \op{RB[3]}.
    \label{lem:rb3correct}
\end{theorem}

\begin{proof}
      \textbf{Validity}:  This follows directly from the structure of the protocol. If a correct process broadcasts a value $m$ then all processes with depth 0 will send \msg{Echo}{m} messages. Therefore all processes with depth 1 will receive such messages from at least one of their quorums. Then all processes with depth 1 will send \msg{Readyafterecho}{1, m} messages. In the next round, all processes with depth 2 will receive such messages from one of their quorums and will deliver the value $m$. Since processes with depth 3 also have depth 2 this means that all processes with depth 3 will deliver $m$, satisfying the validity property. 
    
     \textbf{Consistency}: Suppose a process $p_i$ with depth 3 \textit{dar-delivered} a value $m_i$ and another process with depth 3 $p_j$ \textit{dar-delivered} a value $m_j$. By Lemma~\ref{lem:d3-d1} this means that there exists a process with depth 1 $p_i'$ such that it received \msg{Echo}{m_i} messages from one of its quorums $Q_{i'} \in \mathcal{Q}_{i'}$ and another process with depth 1 $p_j'$ that received \msg{Echo}{m_j} messages from one of its quorums $Q_{j'} \in \mathcal{Q}_{j'}$. Since $Q_i' \cap Q_j' \nsubseteq F$ there is at least one correct process that sent \msg{Echo}{m_i} and \msg{Echo}{m_j} messages. But since a correct process will only send one type of echo messages, we conclude that $m_i=m_j$.

     \textbf{Integrity}:  The variable \op{delivered} guarantees that no correct process delivers multiple times. Moreover, if the sender $p_s$ is correct, the only message $m'$ such that $p_s$ sends a message $\msg{send}{m}'$ is $m$, thus  $m$ was previously \textit{dar-broadcast} by $p_s$.

    \textbf{Totality:}  Suppose that a process with depth 3 $p_i$ has \op{dar-delivered} a value $m$. Then it has obtained $\msg{Readyafterecho}{n, m}$ messages from one of its quorums $Q_i \in mathcal{Q}_i$. Consider any other process with depth 1 $p_j$. Since $p_i$ and $p_j$ are both wise, it holds $F \in \mathcal{F}_i^*$ and $F \in \mathcal{F}_j^*$, which implies $F \in \mathcal{F}_i^* \cap \mathcal{F}_j^*$. Then, the set $K = Q_i \setminus F$ intersects every quorum of $p_j$ and is therefore, a kernel of correct processes for it. Process $p_j$ will receive \msg{Readyafterecho}{n}{m} messages from $K$ and proceed to send a \msg{Readyafterready}{n, m} messages. 

    Since all processes with depth 1 send a  \msg{Readyafterready}{n, m} message, all processes with depth 2 will receive such messages from one of their quorums, prompting them to send a \msg{Readyafterecho}{n+1, m} message. Since all processes with depth 2 will send such messages all processes with depth 3 will receive \msg{Readyafterecho}{n+1, m} from one of their quorums and proceed to \textit{dar-deliver} the value $m$.

\end{proof}

\begin{corollary}
    Algorithm~\ref{alg:rb3} solves \op{RB[$d$]} for any $d \geq 3$
\end{corollary}

In Algorithm~\ref{alg:rb3} and Definition~\ref{def:arbc}, we assume that all correct processes start the execution. This allows us to focus solely on characterizing the depth of the processes that complete the execution. However, in more complex scenarios, such as when protocols are composed sequentially, it is important to also characterize the depth of the processes that start the execution.

For example, consider two protocols executed one after the other. If only processes with depth $d$ complete the first protocol, we cannot assume that all correct processes will initiate the second; we can only rely on those with depth $d$ to do so. This highlights the need to explicitly define both the starting and ending depths of participating processes, especially in composed protocols.

In the case of reliable broadcast, we treat it as a standalone algorithm and thus omit this distinction. In later sections, we will define protocols using two depth parameters: one for the processes that begin the execution, and one for those that finish.

\section{Consensus}
\label{sec:consensus}

We adapt the consensus problem to the depth paradigm and propose an algorithm that solves it for all processes with a minimum depth of at least 10. We follow an approach similar to the one proposed by Abraham~\etal~\cite{DBLP:conf/podc/AbrahamBY22}, who use a round-based algorithm composed of a binding crusader agreement and a common coin.
However, for reasons that will be discussed in Section~\ref{sec:consensus-protocol}, we replace binding crusader agreement with a binding \emph{graded} agreement protocol which allows parties to provide multiple input.
For this purpose, we introduce depth-characterized versions of these primitives in Sections~\ref{sec:cc}~and~\ref{sec:bga}, which we use as subprotocols inside the main consensus algorithm. 

It is at this point that the need to characterize the depth of starting and finishing processes arises. The consensus protocol first calls the binding graded agreement, and then the common coin. Only processes that finish the former will be able to make the call to the latter. Because of this, we need to take into account that not all correct processes can start executing a protocol. We introduce two parameters to characterize the depth of protocols: $d'$ to refer to the depth of processes that start the execution, and $d$ to refer to the depth of processes that finish the execution. We will say that an algorithm works if, in case all processes with depth $d'$ start the execution, all processes with depth $d$ satisfy the required properties. In case we assume that all correct processes start (i.e., $d'=0$), as is the case for standalone protocols (like consensus and reliable broadcast), we will ignore the $d'$ parameter.

\subsection{Common Coin}
\label{sec:cc}

The existing asymmetric common coin protocol proposed by Alpos~\etal~\cite{DBLP:journals/dc/AlposCTZ24} has issues that motivate the development of a new protocol. It is not a purely asymmetric protocol, as each process first constructs the tolerated system, which is a symmetric fail-prone system, and executes the protocol based upon it. This means that the protocol is actually symmetric and is not suitable for asymmetric applications. We address this by first redefining the problem to fit the depth paradigm, shown in Definition~\ref{def:cc}, and propose a modified version of the scheme introduced by Alpos~\etal~\cite{DBLP:journals/dc/AlposCTZ24}, which solves the problem for all processes with depth at least $d+1$, as long as all processes with depth $d$ start the execution. The parameter for all processes starting the protocol is needed since it will be used as part of a pipeline of protocols, which means that only the processes that finish the preceding protocol are guaranteed to start executing the common coin.

\begin{definition}[Depth-characterized asymmetric common coin]\label{def:cc}
  A protocol for \emph{common coin}, for processes with depth $d'$ and $d$, defined by the events \textit{release-coin(round)} and \textit{output-coin(c)}, shortened as \op{CC}$[d', d]$, satisfies the following properties:
\begin{description}
\item[Termination:] If all processes with depth $d'$ call \textit{release-coin}, then every process
  with depth $d$ eventually executes \textit{output-coin}.
  
\item[Unpredictability:] No process has information about the value of the coin before at least one process with depth $d'$ has called \textit{release-coin}.

\item[Matching:] If a process with depth $d$ executes \textit{output-coin(\var{c})} and another process with depth $d$ executes \textit{output-coin(\var{c}$'$)}, then $\var{c}=\var{c}'$.

\item[No bias:] The distribution of the coin is uniform over $\mathcal{B}$.

\end{description}
\end{definition}

Algorithm~\ref{alg:acc} provides a solution to the common coin taking any value $d'$ and $d=d'+1$. In our consensus construction, we require \op{CC[8, 9]} in particular. The common coin we design is strong and it follows the approach of Rabin~\cite{DBLP:conf/focs/Rabin83} and assumes that coins are predistributed by a trusted dealer. The scheme uses Benaloh-Leichter~\cite{DBLP:conf/crypto/Leichter88} secret sharing, such that the coin is additively shared within every quorum.  The dealer shares one coin for every possible round of the protocol and when a process receives a share it can verify that it was indeed created by the dealer. This can be done using standard digital signatures but is outside the scope of this paper. Alpos~\etal~\cite{DBLP:journals/dc/AlposCTZ24} follow the same approach in their scheme. Theorem~\ref{theorem:acc} proves the correctness of the algorithm.

\begin{algo*}[tbh]
\vbox{
\small
\begin{numbertabbing}\reset
  xxxx\=xxxx\=xxxx\=xxxx\=xxxx\=xxxx\=MMMMMMMMMMMMMMMMMMM\=\kill
  \textbf{State} \label{}\\
  \> $\mathbb{Q}$: asymmetric quorum system \label{} \\
  \> \parbox[t]{0.87\linewidth} {$\var{share}[Q][j]$: if $p_i \in Q$, this holds the share received from $p_j$ for quorum $Q$; initially $\nil$} \label{} \\
  \\
  \textbf{upon event} \(\op{release-coin()}\) \textbf{do} \label{line:release-coin-begin}\\
  \> \textbf{for all} \(Q \in \mathcal{Q}_j\) where \(\mathcal{Q}_j \in \mathbb{Q}\) \textbf{such that} $p_i \in Q$ \textbf{do} \label{line:pi-in-G} \\
  \>\> let $s_{iQ}$ be the share of $p_i$ for quorum $Q$ \label{} \\
  \>\> send message \msg{share}{s_{iQ}, Q} to \(p_j\) to all $p_j$ in $\mathcal{P}$\label{line:release-coin-end} \\
  \\
  \textbf{upon} receiving a message \msg{share}{s, Q} from \(p_j\)
  \textbf{such that} $ p_j \in Q$ \textbf{do} \label{line:receive-share-begin} \\
  \> \textbf{if} \(\var{share}[Q][j] = \bot\) \textbf{then}  \label{}\\
  \>\> \(\var{share}[Q][j] \gets s\) \label{line:receive-share-end}\\
  \\
  \parbox[t]{0.87\linewidth} { \textbf{upon } exists $Q \in \mathcal{Q}_i$ \textbf{such that}
    for all $j$ with $p_j \in Q$, it holds \(\var{share}[Q][j] \neq \nil \) \textbf{do}} \label{line:output-coin-begin}\\
  \> $s \gets \sum_{j: p_j \in Q} \var{share}[Q][j]$ \label{line:compute-s}\\
  \> \textbf{output} $\op{output-coin(s)}$ \label{line:output-coin-end}\\[-5ex]
\end{numbertabbing}
}
\caption{Asymmetric common coin  (code for~$p_i$)}
\label{alg:acc}
\end{algo*}

\begin{lemma}
    Consider an execution with faulty processes $F$, a process $p_i$ such that $\op{depth(}p_i\op{)} \geq 1$ and one of its quorums $Q_i \in \mathcal{Q}_i$. For any other process $p_j$ with $\op{depth(}p_j\op{)} \geq 1, \exists K_j \in \mathcal{K}_j$ such that $K_j \subseteq Q_i \setminus F$.
    \label{lem:q-always-k}
\end{lemma}
\begin{proof}
    Since $p_i$ and $p_j$ both have depth at least 1, it holds $F \in \mathcal{F}_i^*$ and $F \in \mathcal{F}_j^*$. This implies $F \in \mathcal{F}_i^* \cap \mathcal{F}_j^*$. Then, the set $Q_i \setminus F$ intersects every quorum of $p_j$ by the quorum consistency property, and therefore contains a kernel for $p_j$.
\end{proof}

\begin{theorem}
    \label{theorem:acc}
Algorithm~\ref{alg:acc} implements \op{CC[d, d+1]}. 
\end{theorem}
\begin{proof}
    \textbf{Termination: } Consider a process $p_i$ with depth $d+1$. There exists a quorum $Q \in \mathcal{Q}_i$ such that all its members have depth at least $d$. Since we assume that they all start the protocol, they will release the coin and send their share to all the processes in $\mathcal{P}$. Therefore, $p_i$ will receive all the coin shares for at least the quorum $Q$ and will be able to output the coin value. 

    \textbf{Unpredictability: } Assume a process $p_i$ with depth $d+1$ outputs the coin. This implies the existence of a set $Q \in \mathcal{Q}_i$ such that each member of $Q$ sent a \str{share} message. Since $p_i$ has a quorum $Q_d \in \mathcal{Q}_i$ such that $\forall p \in Q_i, \op{depth(}p\op{)}\geq d$ and $Q \cap Q_d \neq \emptyset$, it follows that at least one process with depth d has released the coin.

    \textbf{Matching: } Follows from the fact that the value of the coin for each round is predetermined by the dealer.

    \textbf{No bias: } Follows from the fact that the value of the coin for each round is predetermined by the dealer.
\end{proof}

\subsection{Binding Graded Agreement}
\label{sec:bga}

Following the standard pattern, we will build the consensus protocol relying on an intermediate agreement protocol. We will use a binding graded agreement protocol based on the justified proxcensus protocol of Kamp~\cite{DBLP:journals/iacr/Kamp25}. Definition~\ref{def:bga} presents the properties that such a primitive must satisfy to reflect that only a fraction of correct processes might start the execution. This is of use in the round-based context where this primitive will be used, as we can only guarantee that processes with a certain depth will be alive in consensus rounds other than the first. 

Our depth-characterized asymmetric binding graded agreement takes as input one or multiple values and produces as output a value with a corresponding grade. A correct process may start the protocol with a given input and later provide an additional input to the protocol. This detail is crucial for the liveness of the consensus protocol, which will be explained later. 

In contrast to binding crusader agreement, which outputs a value in $\{0,\bot,1\}$, our binding graded agreement additionally outputs a grade between $\{0, 1, 2\}$ which signifies the confidence of the process in the decision. The possible outputs are $(1, 2), (1, 1), (\bot, 0), (0, 1), (0, 2)$ and the protocol must guarantee that the outputs of any two processes with a certain depth are the same or next to each other in the previously provided order. This is made explicit in the graded agreement property. 

\begin{definition}[Depth-characterized asymmetric binding graded agreement]\label{def:bga}
  A protocol for \emph{binding graded agreement} with depths $d'$ and $d$, defined through the events \textit{bga-propose(v)} and \textit{bga-decide(v,g)}, shortened as \op{BGA}$[d', d]$, satisfies the
  following properties:
\begin{description}
\item[Strong Termination:] If there exists a value $v$ such that for each process $p_i$ with depth $d'+1$, a kernel $K_i\in \mathcal{K}_i$ of correct processes inputs $v$, then every process with depth $d$ eventually \textit{bga-decides}.

\item[Validity:] 
If a process with depth $d'$ \textit{bga-decides} on a value different from $(v, 2)$, then a process with depth $d'$ \textit{bga-proposed} $1-v$.

\item[Graded Agreement:] If two processes with depth $d'$ \textit{bga-decide} $(v,g)$ and $(v',g')$, then 
1) $\vert g - g' \vert \leq 1$, 
 2) $v = \bot$ if and only if $g = 0$, and 
3) if $\min(g, g') > 0$ then $v = v'$.

\item[Binding:] Let time $\tau$ be the first time at which there is a process with depth $d$ that \textit{bga-decides}. There is a value $b \in \{0, 1\}$, defined at time $\tau$, such that no process with depth $d$ \textit{bga-decides} $(1-b,\cdot)$ in any extension of this execution. 

\end{description}
\end{definition}

We note that strong termination implies that if all correct processes \textit{bga-propose} a value, then every process with depth $d$ eventually \textit{bga-decides} (Lemma~\ref{lem:kernel-same-value}). However, we need the stronger requirement for composition, as processes with depth $0$ cannot be guaranteed to obtain an input whenever this depends on the output of another protocol.
Algorithm~\ref{alg:bga} presents a solution to this problem for $d'=2$ and $d=6$. We prove its correctness in Theorem~\ref{theorem:bga}. Lemmas~\ref{lem:kernel-same-value},~\ref{lem:rb-consistency-d2},~and~\ref{lem:bga-valid-deliver} prove helper results which are of use in the proof of the main lemma.

To simplify the protocol description, we introduce the following notion of a valid \op{dar-delivery} of a vote.
\begin{definition}[Valid \op{dar-deliver}]
When we say that the $\op{dar-delivery}$ of a vote message $\msg{vote}{\ell,\var{v}_j,Q_j}$ is \emph{valid} in the view of a process $p_i$, we mean that
if $\ell = 0$, then for each $p_k \in Q_j$ a message $\msg{quorum}{\var{v}_j}$ from $p_k$ was $\op{dar-delivered}$.
Otherwise, for $\ell > 0$ we require that for each $p_k \in Q_j$ a message $\msg{vote}{\ell-1,\var{v}_k,Q_k}$ from $p_k$ was validly $\op{dar-delivered}$ such that the average of the votes included in the set $\{v_k \mid p_k \in Q_j\}$ is equal to $v_j/2$.
\end{definition}
In other words, we require that $p_i$ should receive messages that would cause it to vote for $v_j$ if $Q_j$ was its own quorum.
 
The usage of reliable broadcast for these messages guarantees that if a process with some depth delivers a message, the same message will be eventually delivered by any other process with the appropriate depth. This prevents faulty processes from sending incorrect information to a subset of the processes. 
Looking ahead, we show in Lemma~\ref{lem:bga-valid-deliver} that the set of all votes $v$ and $v'$ that belong to the same round and are validly $\op{dar-delivered}$ in the view of a process with depth 2: $\vert v - v'\vert \leq 1$. 
Similarly, we show in Lemma~\ref{lem:bga-validity} that if no party with depth 2 inputs a value different from $v$, then all vote messages that are validly \op{dar-delivered} at parties with depth 2 in round $\ell$ must include a vote for $v\cdot 2^\ell$.
As such, receiving a validly \op{dar-delivered} vote, corresponds to receiving a vote justified by threshold signatures in the proxcensus protocol by Kamp~\cite{DBLP:journals/iacr/Kamp25}.

\begin{algo*}[tbhp]

\vbox{
\small
\begin{numbertabbing}\reset
  xxxx\=xxxx\=xxxx\=xxxx\=xxxx\=xxxx\=MMMMMMMMMMMMMMMMMMM\=\kill
  \textbf{State} \label{}\\
  \> \parbox[t]{0.87\linewidth} {$\var{decided}\gets [\false, \false, \false, \false]: $ // variable to control the sending of vote messages}\label{} \\
  \\

  \textbf{upon event} \(\op{bga-propose(\var{value})}\) \textbf{do} \label{line:bga-propose}\\
  \> send message $\msg{input}{\var{value}}$ to all \(p_j \in \CP\) \label{line:bga-echo} \\
  \\
  \parbox[t]{0.87\linewidth} {\textbf{upon} receiving a message $\msg{input}{\var{value}}$ from all processes in some $K \in \mathcal{K}_i$ \textbf{do}} \label{line:bga-kernel-echo} \\
  \> send message $\msg{kernel}{\var{value}}$ to all \(p_j \in \CP\) \label{line:bga-send-kernel} \\
  \\
  \parbox[t]{0.87\linewidth} {\textbf{upon} receiving a message $\msg{kernel}{\var{value}}$ from all processes in some $Q \in \mathcal{Q}_i$ \textbf{do}} \label{line:bga-kernel-echo'} \\
  \> $\op{dar-broadcast}(\msg{quorum}{\var{value}})$ \label{line:bga-send-quorum} \\
  \\
  \parbox[t]{0.87\linewidth} {\textbf{upon} $\op{dar-deliver}(\msg{quorum}{\var{value}})$ from each sender $p_j$ in some $Q \in \mathcal{Q}_i$ \textbf{do}} \label{line:bga-quorum-echo} \\
  \> \textbf{if} $\neg \var{decided}[1]$ \textbf{then} \label{line:bga-not-sentecho2} \\
  \>\> $\op{dar-broadcast}(\msg{vote}{1,\var{value},Q})$ \label{line:bga-send-vote-1} \\
  \>\> $\var{decided}[1] \gets \true$ \label{} \\
  \\

  \parbox[t]{0.87\linewidth} {\textbf{upon} valid $\op{dar-deliver}(\msg{vote}{\ell,\var{value}_j,Q_j})$ from each sender $p_j$ in some $Q \in \mathcal{Q}_i$ and $\neg \var{decided}[\ell+1]$ \textbf{do}} \label{line:bga-quorum-vote} \\
  \>$\var{decided}[\ell+1]\gets \true$ \label{}\\
  \> \parbox[t]{0.87\linewidth} {$\var{value}\gets 2\cdot\op{mean}({\{ \var{value}_j \mid p_j \in Q\}})$  //set with the different values received in the quorum}\label{line:bga-set-new-value}\\
  \> \textbf{if} $\ell<3$ \textbf{then}\label{}\\
  \>\> $\op{dar-broadcast}(\msg{vote}{\ell+1,\var{value}, Q})$ \label{}\\
  \> \textbf{else} \label{}\\
  \>\> \textbf{if} $\var{value}=0$\label{}\textbf{then}\\
  \>\>\> \op{bga-decide}(0,2)\label{}\\
  \>\> \textbf{elif} $\var{value}=1$\label{}\textbf{then}\\
  \>\>\> \op{bga-decide}(0,1)\label{}\\
  \>\> \textbf{elif} $\var{value}=7$\label{}\textbf{then}\\
  \>\>\> \op{bga-decide}(1,1)\label{}\\
  \>\> \textbf{elif} $\var{value}=8$\label{}\textbf{then}\\
  \>\>\> \op{bga-decide}(1,2)\label{}\\
  \>\> \textbf{else}\label{l:gda2}\\
  \>\>\> $\op{bga-decide}(\bot,0)$\label{}
\end{numbertabbing}
}
\caption{Asymmetric binding graded agreement (code for~$p_i$)}
\label{alg:bga}
\end{algo*}

\begin{lemma}

Given an asymmetric quorum system $\mathbb{Q}$, a set of faulty processes $F$, and any assignment $M: \mathcal{P} \rightarrow \{0, 1\}$ of a value in $\{0, 1\}$ to each process in $\mathcal{P}$. There is a value $v \in \{0, 1\}$, such that each process $p_i$ with depth 1 has a kernel $K \in \mathcal{K}_i$ such that $K \subseteq \mathcal{P} \setminus F$  and $M(p_j) = v$ for all $p_j \in K$.

    \label{lem:kernel-same-value}
\end{lemma}
\begin{proof}
    Given the assignment $M$, we will use the notation $\mathcal{P}\setminus F = S_0 \cup S_1$ where $S_0$, respectively $S_1$, are the processes in $\mathcal{P} \setminus F$ that have 0, respectively 1, in $M$. Suppose that there is no value $v$ such that there is a kernel of correct processes for all processes with depth at least 1 contained in $S_v$. Then, there is a process with depth 1 $p_j$ and a quorum $Q_j\in \mathcal{Q}_j$ such that $Q_j \cap S_0 = \emptyset$.  This also implies that there is a process with depth 1 $p_k$ and a quorum $Q_k\in \mathcal{Q}_k$ such that $Q_k \cap S_1 = \emptyset$. 

    Then, $Q_j\cap Q_k \subseteq F$, but since $p_j$ and $p_k$ have depth at least 1, this breaks the consistency property. Thus, they have to intersect either in $S_0$ or $S_1$ and there must be value $v$ such that every process with depth 1 has at least one kernel contained in $S_v$. 
\end{proof}

\begin{lemma}
    In Algorithm~\ref{alg:rb3}, if a process with depth 3 exists and two processes $p_i, p_j$ with depth 2 \textit{dar-deliver} values $v_i, v_j$ respectively, then $v_i=v_j$. 
    \label{lem:rb-consistency-d2}
\end{lemma}

\begin{proof}
    We assume the existence of a process $p_3$ with depth 3 with a quorum $Q_2 \in \mathcal{Q}_3$ such that all processes in $Q_2$ have depth 2. 
    Process $p_i$ delivers value $v_i$ upon obtaining  \msg{readyafterecho}{r, v_i} messages from a quorum $Q_i \in \mathcal{Q}_i$. Process $p_j$ does the same upon obtaining  \msg{readyafterecho}{r', v_j} messages from a quorum $Q_j \in \mathcal{Q}_j$. Since $Q_i\cap Q_2 \neq \emptyset$ and $Q_j \cap Q_2  \neq \emptyset$, there must be at least one process in both intersections. Denote by $p_i'$ and by $p_j'$ the processes in these intersections, respectively. 

    By Lemma~\ref{lem:d1echo-d1echo} there are processes $p_i'', p_j''$ with depth 1 that sent a \msg{Readyafterecho}{r-1, m}, \msg{Readyafterecho}{r'-1, m} message. By repeatedly applying Lemma~\ref{lem:d1echo-d1echo} we reach processes $p_i''', p_j'''$ with depth 1 that sent a \msg{Readyafterecho}{1, m} message. We now apply Lemma~\ref{lem:d1echo} to obtain that $\exists Q_{i'''} \in \mathcal{Q}_{i'''}$, respectively $Q_{j'''} \in \mathcal{Q}_{j'''}$,  such that all processes in $Q_{i'''}, Q_{j'''}$ sent a \msg{Echo}{v_i}, \msg{Echo}{v_j} respectively message. Since $Q_{i'''} \cap Q_{j'''} \subseteq \mathcal{P}\setminus F$, there must be a correct process that sent echo messages for $v_i$ and $v_j$, but since a correct process only sends echo messages for one value, we obtain that this is not possible. Therefore, $v_i$ and $v_j$ must be equal. 
\end{proof}

\begin{lemma}
    If a process with depth 2 \textit{bga-decides} on a value different from $(v, 2)$, then a process with depth 2 \textit{bga-proposed} $1-v$. 
\label{lem:bga-validity}
\end{lemma}
\begin{proof}

    In what follows, we will assume the existence of a process $p_8$ with depth at least 8. This assumption holds since we are providing a binding graded agreement for processes with depth at least 8. 

    We will prove the contrapositive, that is, if no process with depth 2 proposed $1-v$ then no process with depth 2 decides on a value different from $(v, 2)$. 
     
     We know that if a process with depth 2 \textit{bga-proposes} a value, this value will be $v$. Then, for all processes with depth 3 that receive a kernel of \msg{input}{v'} messages and send a \msg{kernel}{v'} message (line~\ref{line:bga-kernel-echo}), $v'$ must be equal to $v$. This follows from the fact that each kernel for a process with depth at least 3 contains at least one process with depth 2. 
     
     For all processes with depth at least 1 that receive a quorum of \msg{kernel}{v'} values (line~\ref{line:bga-kernel-echo'}), $v'=v$ since from the quorum intersection property it follows that there must be at least one process with depth 3 in that quorum and all of them sent \msg{kernel}{v'}.  

     For all processes with depth 1 that receive a quorum of \msg{quorum}{v'} messages (line~\ref{line:bga-quorum-echo}), the same argument applies and $v'$ must be equal to $v$. Then, all will send \msg{vote}{1, v, Q} messages. 

    For all processes with depth 2 that reliably deliver a quorum of vote messages for round 1, it must be that all vote messages are associated with value $v$. This follows from the \textit{valid} function check before accepting the messages (line~\ref{line:bga-quorum-vote}). Suppose that a process with depth 2 accepts a quorum $Q$ of vote messages for round 1 containing both values $v$ and $v'\neq v$. Then, it must have accepted at least one \msg{vote}{1, v', Q_j} message coming from a process $p_j \in Q$. Since there is a \textit{valid} function check for all these messages, to accept it, the process must have also received \msg{quorum}{v'} messages from all process $p_k\in Q_j$ before. By quorum intersection, there must be at least one process with depth 1 in $Q_j$, and as was shown before, it is not possible for such a process to send a quorum message for something else than $v$. From Lemma~\ref{lem:rb-consistency-d2} we also know that if two processes with depth 2 deliver messages, they must be the same. Then, any process that receives a quorum of \str{vote} messages for round 1, will send a \msg{vote}{2,\var{value},Q_i} with $\var{value}=2\cdot v$. The new assignment for value follows from line~\ref{line:bga-set-new-value}.
    
    The same argument follows for the \str{vote} messages in rounds 2 and 3. Therefore, all vote messages sent by processes with depth 2 in round 2 are done with value $4\cdot v$ and the final value set when receiving a quorum of \str{vote} messages for round 3 is $8 \cdot v$. If $v=0$, any process that \textit{bga-decides} does so with $(0, 2)$. The same argument follows if $v=1$, any process that \textit{bga-decides} does so with $(1, 2)$. This proves the lemma.

\end{proof}

\begin{lemma}

    Given two depth 2 processes $p_i$ and $p_j$ in Algorithm~\ref{alg:bga}, if $p_i$ validly \op{dar-delivers} a message $\msg{vote}{\ell,v,Q}$ and $p_j$ validly \op{dar-delivers} a message $\msg{vote}{\ell,v',Q'}$, then $\vert v'-v\vert\leq 1$, assuming the existence of a process with depth 3. 
    \label{lem:bga-valid-deliver}
\end{lemma}
\begin{proof}
    Consider the first round of \str{vote} messages, a process with depth 2 $p_i$ validly \textit{dar-delivers} a message \msg{vote}{1,v,Q} upon receiving the message itself and a \msg{quorum}{v} message from each process in $Q$. Since the only input values of correct processes are 0 or 1, it is not possible for any \str{quorum} message to contain a value $v$ different from 0 or 1. Therefore, any process $p_i$ that validly \op{dar-delivers} a message \msg{vote}{1,v,Q} will do so with $v=0$ or $v=1$. 
    So, for any two processes with depth 2 $p_i$, $p_j$, that validly \textit{dar-deliver} \str{vote} messages $\msg{vote}{\ell,v,Q}$, $\msg{vote}{\ell,v',Q'}$, it must hold that $|v'-v|\leq 1$.

    Assume that the statement holds for $\ell-1\geq 1$, namely that if 
    $\msg{vote}{\ell-1,v,Q}$, $\msg{vote}{\ell-1,v',Q'}$ are validly $\op{dar-delivered}$ by any pair of processes with depth 2, then $\vert v - v'\vert \leq 1$.
    
    To prove the statement for $\ell \in \{2, 3\}$ we first observe that a depth 2 process will only validly \op{dar-deliver} a message $\msg{vote}{\ell,v^*,Q^*}$, if messages of the form $\msg{vote}{\ell-1,v,Q}$ were validly \op{dar-delivered} from every sender $p$ in $Q^*$ and $v^*$ is twice the average of the set of included votes (line~\ref{line:bga-set-new-value}).
    By induction, all such votes are either all on the same value $v$, or split between two values $v$ and $v'=v+1$ \emph{for all parties of depth 2}.
    The first case is straightforward, as $v^*$ must be equal to $2v$ for all depth 2 processes who validly \op{dar-deliver} a message $\msg{vote}{\ell,v^*,Q^*}$.
    In the latter case, all depth 2 processes who validly \op{dar-deliver} a message $\msg{vote}{\ell,v^*,Q^*}$ only do so
    if $v^*$ equals $2v$, $v + v' = 2v+1$, or $2v' = 2v+2$.
    We only need to rule out the possibility of $\msg{vote}{\ell,2v,Q^*}$ being validly \op{dar-delivered} by $p_i$ and $\msg{vote}{\ell,2v',Q^{**}}$ being validly \op{dar-delivered} by $p_j$ for any depth 2 processes $p_i$ and $p_j$ as the remaining combinations satisfy the statement.
    For this, we use the fact that all depth 2 processes cannot \op{dar-deliver} inconsistent messages (Lemma~\ref{lem:rb-consistency-d2}) and that $Q^*$ and $Q^{**}$ intersect on at least one party. So, in particular it is impossible for depth 2 processes $p_i$ and $p_j$, that $p_i$ validly \op{dar-delivers} $\msg{vote}{\ell,2v,Q^{*}}$ 
    and $p_j$ validly \op{dar-delivers} $\msg{vote}{\ell,2v',Q^{**}}$ as they would each have \op{dar-delivered} $\msg{vote}{\ell-1,\cdot,\cdot}$ from all parties in $Q^*$ and $Q^{**}$, respectively, and for all parties in 
    $Q^* \cap Q^{**} \neq \emptyset$ the messages would be inconsistent as the votes included would need to be $v$ for $p_i$ and $v'=v+1$ for $p_j$.

\end{proof}

\begin{theorem}
    \label{theorem:bga}
Algorithm~\ref{alg:bga} solves depth-characterized asymmetric binding graded agreement with $d'=2$ and $d=8$. 
\end{theorem}

\begin{proof}
    We structure the proof property by property.
    
    \textbf{Strong Termination: } 
    Assume that there exists a value $v$ such that for each process $p_i$ with depth 3, there is an honest kernel $K\in\mathcal{K}_i$ that sends a \msg{input}{v} message. Therefore, all processes with depth 3 receive a message \msg{input}{v} from one of their kernels and send a message \msg{kernel}{v} (line~\ref{line:bga-send-kernel}). From this, all processes with depth 4 will receive a quorum of \msg{kernel}{v} messages and reliably broadcast a message \msg{quorum}{v} (line~\ref{line:bga-send-quorum}). Upon receiving \msg{quorum}{v} messages from a quorum $Q$, which happens for all processes with depth 5, the processes will reliably broadcast a message \msg{vote}{1, v, Q} (line~\ref{line:bga-send-vote-1}). Since this is only done once, by the usage of the $\var{decided[1]}$ variable, note that it is possible for one process with depth 5 to send this vote with $v=0$ and for another process with depth 5 to do this with $v=1$. This follows from the possible existence of kernels for both 0 and 1 for different processes among the processes that start the execution. 

    All processes with depth 6 will then receive votes with level 1 from one of their quorums and send a \msg{vote}{2, \var{value}', Q} message, where $\var{value}'$ is determined using line~\ref{line:bga-set-new-value}. Therefore, since all processes with depth 6 send this vote message, all processes with depth 7 will receive them from at least one of their quorums. Subsequently, they send a vote message for level 3, which will be received by all processes with depth 8, which will \str{bga-decide} a value and terminate the execution.

    \textbf{Graded Agreement: } Assume that processes $p_i$ and $p_j$ $\op{bga-decide}(v,g)$ and $\op{bga-decide}(v',g')$ respectively.

    \begin{enumerate}
        \item[1)] Note that $p_i$ computes $\var{value}$ as $2\cdot\op{mean}({\{ \var{value}_k \mid p_k \in Q_i\}})$ from the values $\var{value}_k$ from a quorum $Q_i\in\CQ_i$ of messages $\msg{vote}{3,\var{value}_k,Q_k}$ (Line~\ref{line:bga-quorum-vote}). Analogously, $p_j$ computes $\var{value}'$ as $2\cdot\op{mean}({\{ \var{value}'_k \mid p'_k \in Q_j\}})$ from the values $\var{value}_k$ from a quorum $Q_j\in\CQ_j$ of messages $\msg{vote}{3,\var{value}_k,Q'_k}$. Since we assume a process with depth 8, Lemma~\ref{lem:bga-valid-deliver} implies that $\vert \var{value}'-\var{value}\vert \leq 1$, and the fact that $Q_i\cap Q_j$ implies that at most one of the quorums sent a single value, thus the $\var{value}$ and $\var{value}'$ computed by $p_i$ and $p_j$ satisfy $\vert \var{value}-\var{value}'\vert \leq 1$, hence $\vert g-g'\vert \leq 1$. 
        
        \item[2)] Follows from line~\ref{lem:bga-validity}.

        \item[3)] W.l.o.g. assume that $\var{value}'<\var{value}$. Following the reasoning of 1), the values $\var{value}$ and $\var{value}'$ computed by $p_i$ and $p_j$ in the last round of the algorithm satisfy $\vert \var{value}-\var{value}'\vert \leq 1$, thus if $\min\{g,g'\}>0$, this means that  then $v=v'$.
    \end{enumerate}

    \textbf{Validity: } Follows from Lemma~\ref{lem:bga-validity}, as we assume the existence of a process with depth 8.

    \textbf{Binding:} Consider the first process $p_8$ with depth $d=8$ that 
    $\op{dar-delivers}$ messages $\msg{vote}{3,\var{value}_j,Q_k}$ with every sender $p_k$ of one of its quorums $Q_i$ and $\op{bga-decides}(v,g)$. Define this time as $\tau$.
    If $v\neq \bot$, graded agreement implies binding, so assume $v=\bot$.

    Since $v=\bot$, then $\var{value}\in \{2,...,6\}$ and it is computed as $\var{value}=2\cdot\op{mean}({\{ \var{value}_k \mid p_k \in Q_i\}})$. 
    
    If $\var{value}$ is even (2,4,6), then $\var{value}_k=\frac{\var{value}}{2}$ for every $p_k\in Q_i$. Lemma~\ref{lem:bga-valid-deliver} guarantees that any other process $p_j$ with depth 8 that $\op{dar-delivers}$ messages $\msg{vote}{3,\var{value}'_k,Q'_m}$ with every sender $p_m$ of one of its quorums $Q_j$ satisfies $\vert \frac{\var{value}}{2}-\var{value}_m\vert \leq 1$. And quorum intersection with agreement of reliable broadcast implies that there exists $\var{value}_m=\frac{\var{value}}{2}$. Thus, process $p_j$ computes $\var{value}'=2\cdot\op{mean}({\{ \var{value}_m \mid p_m \in Q_j\}})=\var{value}\pm 1$. Thus, any process $p_j$ with depth 8 returns:
    \begin{itemize}
        \item If $\var{value}=2$, $p_j$ $\op{bga-decides}(0,1)$ or $\op{bga-decides}(\bot,0)$.
        \item If $\var{value}=4$, $\op{bga-decides}(\bot,0)$.
        \item If $\var{value}=6$, $p_j$ $\op{bga-decides}(1,1)$ or $\op{bga-decides}(\bot,0)$.
    \end{itemize}

    If $\var{value}$ is odd (3,5), the same reason from above, so any process $p_j$ with depth 8 computes $\var{value}'=\var{value}\pm 1$. But, this time, $p_j$ $\op{bga-decides}(\bot,0)$. And the binding property is fulfilled.

    We have shown that, if we take $d'=2$ and $d=8$, the algorithm satisfies all required properties. Therefore, Algorithm~\ref{alg:bga} solves asymmetric binding graded agreement with $d=8$ and $d'=2$. 
\end{proof}

\subsection{Consensus}
\label{sec:consensus-protocol}

Definition~\ref{def:aconsensus} introduces the concept of depth to the asymmetric consensus problem. We propose a consensus algorithm using the round-based paradigm of Abraham~\etal~\cite{DBLP:conf/podc/AbrahamBY22}. Each round consists of a call to the binding graded agreement and common coin primitives presented earlier. If the values obtained from the two primitives match, then the processes deliver it; otherwise, another round is executed.

We make a small addition to the common coin and binding graded agreement calls. When one of these sub-protocols is triggered, it is done so with a \var{round} tag. This is used to keep track of the round that is running and does not interfere with any details of the sub-protocols. When the protocol call returns, it will also include the tag with which it was called.

\begin{definition}[Depth-characterized asymmetric Byzantine consensus]\label{def:aconsensus}
  A protocol for \emph{asymmetric Byzantine consensus}, shortened as \op{C[d]}, defined through the events \op{c-propose(\var{val})} and \op{c-decide(\var{val})}, satisfies:

\begin{description}
\item[Termination:] Every process with depth $d$ eventually $\op{c-decides}$ with probability~$1$.

\item[Validity:] If all correct processes \op{c-propose} the same value \var{val}, then no processes with depth $d$ \op{c-decides} $1-\var{val}$.

\item[Agreement:] No two processes with depth $d$ \op{c-decide} different values.

\end{description}
\end{definition}

A difference with respect to symmetric round-based protocols is that not all processes will finish all the sub-protocol calls. Only processes with depth 8 can complete the graded agreement, and only those with depth 9 also finish the subsequent common coin. If we follow a naive approach, in which only those processes that completed the previous round participate in the next, the required depth increases with the number of rounds.

We address this fact with the strong termination property of binding grading agreement, which guarantees liveness, even when not every process starts the execution of the protocol. Without this element, after each round, liveness could only be guaranteed for processes with higher and higher depths. Since such a protocol can have executions with infinite length~\cite{DBLP:journals/jacm/FischerLP85}, we would only be able to provide guarantees for processes with infinite depth, which goes against our objectives here.

A process may execute multiple rounds in parallel. This is enabled by the usage of arrays (\var{\propose}, \var{\proposed}) to store the process states separately for each round. This design ensures that the process can continue sending messages for every round it has initiated, which may be needed by other processes in the system. 

Algorithm~\ref{alg:aconsensus} presents a solution to the asymmetric consensus problem for all processes with depth 9. It relies on the binding graded agreement and common coins previously introduced. Theorem~\ref{theorem:aconsensus} proves that it satisfies the properties described in Definition~\ref{def:aconsensus}. Lemmas~\ref{lem:c-cc-matching-2},~\ref{lem:c-d2-start-d7-end},~and~\ref{lem:c-d7-end-d2-start} show helper results that are used in the proof of the theorem. Since we want to prove the properties for all processes with depth 9, in the proofs, we will assume that there is at least one such process.

\begin{algo*}[tbhp]
\vbox{
\small
\begin{numbertabbing}\reset
  xxxx\=xxxx\=xxxx\=xxxx\=xxxx\=xxxx\=MMMMMMMMMMMMMMMMMMM\=\kill
  \textbf{State} \label{}\\
  \> \parbox[t]{0.87\linewidth} {$\var{bga-out} \gets [\ ]$: array that stores bga decisions for each round, initially empty} \label{} \\
  \> \parbox[t]{0.87\linewidth} {${\propose} \gets [\false, \false]^\infty$: $\propose[r][v]$ indicates if $v$ should be proposed for round $r$ } \label{} \\
  \> \parbox[t]{0.87\linewidth} {${\proposed} \gets [\false, \false]^\infty$: $\proposed[r][v]$ indicates if $v$ has been proposed for round $r$ } \label{} \\
  \\

  \textbf{upon event} \op{c-propose}(\var{value}) \textbf{do} \label{}\\
  \> ${\propose[{value}]}[1] \gets \true$ \label{} \\
  \\
  
  \parbox[t]{0.87\linewidth} {\textbf{upon } exists \var{round} and \var{value} where $\propose{[\var{round}][value]}$ and $\neg \proposed[[\var{round}]{value}]$ \textbf{do} } \label{line:ba-start-next-round} \\
  \> $\proposed[round][{value}] \gets \true$ \label{} \\
  \> \op{bga-propose(\var{value})} with tag \var{round} \label{line:c-propose} \\

  \\
  
  \parbox[t]{0.87\linewidth} {\textbf{upon } \op{bga-decide(\var{(value, grade)})} with tag \var{round} \textbf{do} // after finishing the BGA protocol}\label{line:c-bga-decide} \\
  \> \var{bga-out}[\var{round}] $\gets$ \var{(value, grade)} \label{line:c-v-gets-value} \\
  \> \op{release-coin()} with tag \var{round} // call the common coin protocol \label{} \\
  \\
  \parbox[t]{0.87\linewidth} {\textbf{upon } \op{output-coin(\var{c})} with tag \var{round} \textbf{do} // after finishing the common coin protocol} \label{line:c-output-coin} \\
  \>  $(\var{value}, \var{grade}) \gets \var{bga-out}[round]$ \label{} \\
  \> \textbf{if} $grade < 2$ \textbf{then} \label{line:c-low-grade} \\
  \>\> $\propose[round+1][{c}] \gets \true$ \label{line:c-v-gets-c} \\
  \>  \textbf{if} $grade > 0$ \textbf{then} // $value \neq \bot$ \label{line:c-high-grade} \\
  \>\> $\propose[round+1][{value}] \gets \true$  \label{} \\
  \> \textbf{if } $grade = 2$ and $\var{c} = value$  \textbf{then}  \label{line:c-decide} \\
  \>\>  \op{c-decide}(value) \label{} \\
  \\[-5ex]
\end{numbertabbing}
}
\caption{Asymmetric binary consensus  (code for~$p_i$)}
\label{alg:aconsensus}
\end{algo*}

\begin{lemma}
    If for all processes with depth at least 3 a kernel of correct processes  \op{bga-propose} the same value $v$ with tag \var{round} in Algorithm~\ref{alg:aconsensus}, 
    then for some value $v'$: all processes with depth at least $9$ \op{bga-propose} $v'$ with tag \var{round+1}.
    \label{lem:c-d2-start-d7-end}
\end{lemma}
\begin{proof}
    Assume that for all processes with depth at least 3 a kernel of correct processes \op{bga-propose} the same value $v$ with tag \var{round} (line~\ref{line:c-propose}), then by the strong termination property of \op{BGA[2,8]}, all processes with depth $8$ \op{bga-decide} on a value and a grade and thus invoke $\op{CC[8,9]}$ with tag \var{round} (line~\ref{line:c-bga-decide}).
    Then all processes with depth $9$ receive $\op{output-coin}$ for some value $c\in \zo$ (line~\ref{line:c-output-coin}) due to termination of $\op{CC[8,9]}$.
    These processes evaluate both the if statement in line~\ref{line:c-low-grade} and line~\ref{line:c-high-grade}. Due to graded agreement of $\op{BGA[2,8]}$ and the overlap on $g=1$ for the two conditions: if a process with depth $9$ has grade $g=2$, then all processes with depth 9 evaluate line~\ref{line:c-high-grade} to true (and propose the same value as $g>0$). Likewise, if a process with depth $9$ has grade $g=0$, then all processes with depth $9$ evaluate line~\ref{line:c-high-grade} to true (and propose the same value due to the matching property of $\op{CC}[8,9]$).
\end{proof}

\begin{lemma}
    If in Algorithm~\ref{alg:aconsensus}, all processes with depth at least $9$ \op{bga-propose} $v$ with tag \var{round} for the same value $v$,
    then these processes form a kernel for all processes with depth 3.
    \label{lem:c-d7-end-d2-start}
\end{lemma}
\begin{proof}
    Assume the existence of at least one process $p_{10}$ with depth 10.
    Assume that all processes with depth at least 9 \op{bga-propose} $v$ with tag \var{round} for the same value $v$.
    Consider a quorum for $p_{10}$ consisting of a subset of these depth $9$ processes. By Lemma~\ref{lem:q-always-k} these are a kernel for all processes with depth $1$. 

\end{proof}

\begin{lemma}
    The matching property of \op{CC}[8,9] holds for all processes with depth $2$.
    \label{lem:c-cc-matching-2}
\end{lemma}
\begin{proof}
    This follows directly from the proof of Theorem~\ref{theorem:acc}, as the proof of the matching property does not rely on the level of a correct process.
\end{proof}

\begin{theorem}
    Algorithm~\ref{alg:aconsensus} solves \op{C[10]}.
    \label{theorem:aconsensus}
\end{theorem}
\begin{proof}
    \textbf{Validity: } We assume that all correct processes (i.e., those with depth 0) start the execution with the same value $\var{val}$. Therefore, they $\op{bga-propose}(\var{val})$ and have no path to $\op{bga-propose}(\var{1-val})$. The validity property of binding graded agreement ensures that every process of depth $8$ $\op{bga-decides}(\var{val,2})$. Every process with depth $9$ receives the output of the coin and sets $\var{propose}[\var{round}][\var{val}]$ to $\true$, regardless of the value of the coin.  Thus, no process with depth 2 $\op{bga-proposes}(1-\var{val})$ and the validity of binding graded agreement, iterating over the same argument, ensures that no process with depth 11 $\op{c-decide}(1-\var{val})$.

    \textbf{Termination \& Agreement: } 
    We prove that every process with depth $9$ eventually $\op{c-decides}(b)$ on a common bit $b$ under the assumption that a process of depth $10$ exists.
    Given a round $r$, Lemma~\ref{lem:c-d2-start-d7-end} and Lemma~\ref{lem:c-d7-end-d2-start} implies that every process with depth 9 eventually invokes $\op{bga-propose}$ some value in round $r$.
    If in round $r$, all processes with depth $9$ propose the same value, then by strong termination of $\op{BGA}[2,8]$, and termination and matching of $\op{CC}[8,9]$: each process $p_i$ with depth $9$ ends up with a common coin $c\in \zo$ and a pair $(v_i,g_i)\in \{0,\bot,1\}\times \{0,1,2\}$. 
    By Lemma~\ref{lem:c-cc-matching-2}, the same holds for all processes of depth $2$, if they get both outputs.
    Additionally, since $\op{BGA}[2,8]$ satisfies binding for processes with depth $8$ and $\op{CC}[8,9]$ satisfies unpredictability for processes with depth $8$: the value of $c$ is unpredictable until a value $b\in\zo$ is fixed in the execution and restricts $v_i\in \{b,\bot\}$ for all processes with depth at least $8$.
    If any process $p_i$ with depth 9 has $g_i>0$, then $b=v_i$ by graded agreement of $\op{BGA}[2,8]$.
    Since $c$ is uniform and independent of $b$: $c$ is equal to $b$ with probability $\frac{1}{2}$, i.e., after $2$ iterations in expectation.
    Whenever this happens, all depth $9$ processes will propose this one value for round $r+1$ (line~\ref{line:c-high-grade} and line~\ref{line:c-low-grade}), and no process with depth $2$ will propose a different value. 
    Let $r'$ be the first round where $c'$ equals $b'$, and observe that this is the only condition under which a correct party may $\op{c-decide}(b')$ (line~\ref{line:c-decide}).
    Now, applying validity of \op{BGA}$[2,8]$ to all instances of \op{BGA}$[2,8]$ with tags $r'+1$ and up, it follows that all processes with depth $9$ \op{bga-decide}$(b',2)$ and no processes with \op{bga-decide} on any other value. 
    Whenever a coin $c''$ is equal to $b'$ again in some round $r''>r'$, all parties with depth $9$ \op{bga-decide}$(b)$.

\end{proof}

\section{Conclusion and Future Work}
\label{sec:conclusion}
Prior to this work, solutions to problems such as reliable broadcast and consensus relied on guild assumptions. Theorem~\ref{th:guild-implies-sym} shows that this assumption effectively reduces to the existence of a symmetric system, constructed from the original assumptions, that provides equivalent or stronger guarantees. In other words, existing solutions are not well-suited for genuinely asymmetric settings.

To overcome this limitation, we introduced a more fine-grained way to characterize asymmetry through the notion of \emph{depth}. We proved that reliable broadcast can be solved for all processes with depth at least 3, and consensus for processes with depth at least 9. This significantly relaxes the requirements compared to prior work by Alpos~\etal~\cite{DBLP:journals/dc/AlposCTZ24}, which established feasibility only for processes with infinite depth (i.e., guild members).

The exact minimum depth required for solving asymmetric reliable broadcast and consensus remains an open problem. For reliable broadcast, the lower bound is either 2 or 3; for consensus, it lies between 1 and 9. Determining these thresholds is crucial, as they capture the precise necessary and sufficient conditions for solving problems under asymmetric trust. Previous work~\cite{DBLP:journals/dc/AlposCTZ24, DBLP:conf/wdag/LiCL23} established sufficiency but left necessity unproven.

Finally, an important direction for future research is to explore whether the concepts and techniques introduced here extend to other asymmetric trust models~\cite{DBLP:conf/wdag/LiCL23, DBLP:conf/wdag/LosaGM19}.

\section*{Acknowledgments}

This work was supported by the Swiss National Science Foundation (SNSF)
under grant agreement Nr\@.~219403 (Emerging Consensus), by the Initiative for
Cryptocurrencies and Contracts (IC3), by the Cryptographic Foundations for
Digital Society, CryptoDigi, DFF Research Project 2, Grant ID
10.46540/3103-00077B, and by the European Union, ERC2023-StG-101116713. Views and opinions expressed are those of the author(s) only and do not necessarily reflect those of the European Union. Neither the European Union nor the granting authority can be held responsible for them.

\bibliographystyle{plainurl}
\bibliography{references, dblpbibtex}

\newpage

\end{document}